%% file: main.tex
\documentclass[11pt,twoside]{article} 

\usepackage{tikz}
\usepackage{enumitem}
\usepackage{algorithm}
\usepackage{algpseudocode}
\usepackage{natbib}
\usepackage{systeme}
\usepackage{bbm}
\usepackage{multirow}
\usepackage{amsthm}
\usepackage{thmtools}
\usepackage{thm-restate}
\usepackage{chngcntr}
\usepackage{graphicx}
\usepackage{amsmath,amssymb,mathtools}
\usepackage[pagebackref,colorlinks=true,urlcolor=blue,linkcolor=blue,citecolor=violet,pdfstartview=FitH]{hyperref}
\usepackage[nameinlink]{cleveref}

\newcommand{\pr}[1]{\text{\bf Pr}\normalfont\lbrack #1 \rbrack}

\newcommand{\ex}[1]{\text{\bf E}\normalfont\lbrack #1 \rbrack}

\def\CUT{\mathsf{CUT}}
\def\CROSS{\mathsf{CROSS}}
\def\LINEAR{\mathsf{LINEAR}}

\newcommand{\budget}{\mathsf{budget}}
\newcommand{\GRC}{\mathsf{GRC}}
\newcommand{\GRone}{\mathsf{GR1}}
\newcommand{\GRtwo}{\mathsf{GR2}}

\def\add{\mathsf{add}}
\def\remove{\mathsf{remove}}
\def\calC{\mathcal{C}}
\def\calR{\mathcal{R}}
\def\calB{\mathcal{B}}

\def\calF{\mathcal{F}}

\def\BS{\mathsf{BinSearch}}
\def\Bshouty{\mathsf{BshoutyCW}}
\def\BMdet{\mathsf{WtdGraphReconDetBM}}

\def\Choi{\mathsf{WtdGraphReconChoi}}

\makeatletter
\newenvironment{subtheorem}[1]{%
  \def\subtheoremcounter{#1}%
  \refstepcounter{#1}%
  \protected@edef\theparentnumber{\csname the#1\endcsname}%
  \setcounter{parentnumber}{\value{#1}}%
  \setcounter{#1}{0}%
  \expandafter\def\csname the#1\endcsname{\theparentnumber.\Alph{#1}}%
  \expandafter\def\csname theH#1\endcsname{thm.\theparentnumber.\Alph{#1}}%
  \unskip\ignorespaces
}{%
  \setcounter{\subtheoremcounter}{\value{parentnumber}}%
  \ignorespacesafterend
}
\makeatother
\newcounter{parentnumber}

\newtheorem{theorem}{Theorem}[section]
\newtheorem{lemma}[theorem]{Lemma}
\newtheorem{proposition}[theorem]{Proposition}

\newtheorem{corollary}[theorem]{Corollary}
\newtheorem{definition}[theorem]{Definition}

\algrenewcommand\algorithmicrequire{\textbf{Input:}}
\algrenewcommand\algorithmicensure{\textbf{Output:}}

\Crefname{algocf}{Algorithm}{Algorithms}

\def\resdeg{\mathsf{resdeg}}
\def\resedge{\mathsf{resedge}}

\begin{document}

\title{Learning Spanning Forests Optimally using CUT Queries \\in Weighted Undirected Graphs}

\author{Hang Liao \thanks{hang.liao.gr@dartmouth.edu, Department of Computer Science, Dartmouth College} \and 
 Deeparnab Chakrabarty \thanks{deeparnab@dartmouth.edu, Department of Computer Science, Dartmouth College}}

\date{}

\maketitle

\begin{abstract}%
	 In this paper we describe a randomized algorithm which returns a maximal spanning forest of an unknown {\em weighted} undirected graph making 
 $O(n)$ $\CUT$ queries in expectation. For weighted graphs, this is optimal due to a result in [Auza and Lee, 2021] which shows an $\Omega(n)$ lower bound for zero-error randomized algorithms. 
 These questions have been extensively studied in the past few years, especially due to the problem's connections to symmetric submodular function minimization. 
We also describe a simple polynomial time deterministic algorithm that makes $O(\frac{n\log n}{\log\log n})$ queries on undirected unweighted graphs and returns a maximal spanning forest, thereby (slightly) improving upon the state-of-the-art.

\end{abstract}

\input{section1_intro}

\input{section2_prelim}

\input{section3_main_alg}
\input{section4_det_alg}

\input{section5_conclusion}

\bibliographystyle{plainnat}
\bibliography{main}

\newpage 

\appendix

\input{sectionA_missing_proofs}


\end{document}

%% file: section1_intro.tex
\section{Introduction}

Learning an unknown graph via queries has been extensively studied for more than two decades. The general setting is this: there is an undirected graph $G = (V,E)$ whose vertices are known but the edge set is unknown, and certain kinds of queries are allowed on this graph. The goal is to reconstruct the graph with as few queries as possible. 
Such active learning questions also have applications in fields such as computational biology~(cf. \citet{GrebinskiK98}), and connections to data summarizations or sketches where the answers to the queries can be thought of as holding the ``relevant information'' about the graph.
More generally, this question falls under the umbrella of {\em combinatorial search} (cf.~\citet{Aigner88,DuH00}) which is a vast area of study that wishes to \emph{``determine an unknown object by means of indirect questions about this object''\footnote{quote from~\citet{GrebinskiK00}}}.

In this paper, we consider {\em cut-query} access to an unknown, undirected, {\em weighted}/multi graph. Every $e\in E(G)$ has an associated $w(e) > 0$. For $e \notin E(G)$, we assume $w(e) = 0$. 
Given a subset $S\subseteq V$ of vertices, let $\partial S$ be the set containing edges with exactly one endpoint in $S$. A $\CUT$ query takes $S\subseteq V$ as input and returns the {\em value} $\sum_{e\in \partial S} w(e)$. 
Instead of focusing on the graph reconstruction question (which has been almost fully resolved; see~\Cref{sec:rel}), we ask in how few queries can one decide if $G$ is connected, or more generally, find a {\em maximal} spanning forest\footnote{We caution the reader that we are not finding a {\em maximum weight} spanning forest with respect to these $w(e)$'s. Finding even a good approximation to the
maximum weight spanning forest can be shown to need $\widetilde{\Omega}(n^2)$ queries whose proof, although not difficult, is out of scope of this paper.}  in $G$. This particular question of connectivity has seen a lot of interest in the recent years~(cf. \citet{RubinsteinSW18,Graur,LeeMS21,AuzaL21,assadi,Apers,ChakrabartyL23}) mainly due to the connections to streaming and sketching, but also due to connections to submodular function minimization (SFM). The cut-function of an undirected graph is a well known (symmetric) submodular function, and such functions can be (non-trivially) minimized using $O(n^3)$ queries (\cite{Quey98}) deterministically, and in $\tilde{O}(n^2)$ queries using randomization (\citet{ChekuriQ}). On the other hand, there is no $\omega(n)$ lower bound known for this question. This has led to recent interest in understanding if for the special case of undirected graph connectivity, can one design $O(n)$-query algorithms, or can a super-linear lower bound be proven.

It is not too hard~(cf.~\citet{Harvey08}, Theorem 5.10) to design  a deterministic $O(n\log n)$ query algorithm to find a spanning forest of an unknown, weighted, undirected graph via mimicking a Prim-style algorithm using a binary-search style idea. To the best of our knowledge (cf. \citet{Apers}, Table 1), this is still the best known result for deterministic algorithms, even for unweighted undirected graphs.
\citet{Apers} gives a randomized, zero-error algorithm for this problem on {\em unweighted} graphs, which makes $O(n)$ queries in expectation. 
However, as argued in~\citet{ChakrabartyL23}, this algorithm used the unweightedness quite crucially, and this latter paper gave an $O(n \log\log n \cdot(\log\log\log n)^2)$-query randomized Monte Carlo algorithm that solved the connectivity question on weighted graphs with constant probability (which is weaker than a zero-error algorithm). It was left open to match \citet{Apers} result
for weighted graphs. The main result of this paper is an affirmative resolution.

\begin{theorem}\label{thm:1}
	Given $\CUT$ query access to an unknown weighted undirected graph $G=(V,E,w)$ with non-negative weights, there is a polynomial time Las Vegas algorithm returning a maximal spanning forest of $G$ that makes $O(n)$ queries in expectation.
\end{theorem}
The query complexity of our algorithm is optimal up to a constant factor; \citet{AuzaL21} prove that any zero-error randomized algorithm to detect whether a graph is connected or not, even when the query model is a much stronger\footnote{They consider the $\LINEAR$ model where one specifies a $\binom{n}{2}$-dimensional query vector $q$ and obtains the answer $\sum_e q(e)w(e)$.} query model than the $\CUT$ query model, 
needs to make at least $\Omega(n)$ many such queries. To the best of our knowledge, it is the only regime of this problem where the upper and lower bounds are tight up to constants, and completes the story for learning maximal spanning forests in weighted undirected graphs using zero-error randomized algorithms.


Our second result is a {\em deterministic} algorithm for undirected, unweighted graphs which makes $O(\frac{n\log n}{\log\log n})$ queries.

\begin{theorem}\label{thm:2}
	Given $\CUT$ query access to an unknown unweighted undirected graph $G=(V,E)$, there is a deterministic algorithm that makes $O(\frac{n\log n}{\log\log n})$ queries and returns a maximal spanning forest in $G$.
\end{theorem}
\noindent
This (slightly) improves upon the $O(n\log n)$-query algorithm mentioned in~\citet{Harvey08}. It is known that Harvey's algorithm can also work with the much weaker ``OR'' query model where one only gets to know if the cut value is zero or positive, and in this model, via a connection to communication complexity of graph connectivity (\cite{HajnalMT88}), an $\Omega(n\log n)$ lower bound is known for deterministic algorithms. Our algorithm shows that with the stronger $\CUT$ queries one cannot obtain the same lower bound, and thus opens up possibilities of much better deterministic algorithms. 

\paragraph{Perspective.} The problem of finding a maximal spanning forest and deciding the connectivity of a graph is a classic algorithmic question which has been studied in many models of computation including dynamic (\cite{KapronKM13,DuanZ17}), streaming (\cite{AhnGM12,NelsonY18}), and parallel (\cite{AndoniSSWZ18,BehneDELM19}) computing. Our first result, 
\Cref{thm:1}, gives a tight understanding of the problem in terms of query complexity, closing the gap between $O(n\log\log n)$ and $O(n)$ left by previous works. 
Quantitatively, this may not seem like a big improvement, but it has qualitative value in (a) being an end-of-the-line study for this problem, (b) ruling out cut functions on weighted undirected graphs as candidate lower bounds for symmetric SFM, and (c) as we explain in the next subsection, leading to a new algorithmic technique which may be helpful in other problems. 
Furthermore, our second result, \Cref{thm:2}, shows that for deterministic algorithms, the simple $O(n\log n)$-query is not optimal, and therefore can lead to more interest in coming up with deterministic strategies 
for solving this problem. 

\subsection{Technical Contribution}\label{sec:tech}

In this section, we give a technical overview of both our results. The maximal spanning forest algorithms in our and all previous papers use the following ``Bor\r{u}vka style'' framework: begin with a collection of $n$ singleton connected components and in phases merge connected components till one gets a maximal spanning forest. What differs is what each phase does.

\paragraph{Randomized Algorithm (\Cref{thm:1}) approach.} A phase in our algorithm is a randomized algorithm that takes input a {\em weighted} graph with $t$ connected components consisting of learnt edges, and then performs queries to discover new edges so that the number of connected components go down to $ct$ for some constant $c < 1$. The whole creativity lies in how to do this using only $O(t)$ queries in expectation. Once we have this, a simple geometric sum gives the $O(n)$ query algorithm in expectation: the algorithm proceeds in $O(\log n)$-rounds making $O(n + cn + c^2n + \ldots) = O(n)$ many queries in all. We note here that~\cite{Apers} obtains the same result on unweighted graphs, but our algorithm is {\em different} and not just a generalization of their algorithm.

To illustrate and underscore our main technical contribution, in~\Cref{subsec:1}, we first focus on a simpler version of the problem which is, in some sense, the first phase. We assume our unknown, {\em weighted}, undirected graph $G = (R\cup B, E)$ is {\em bipartite}\footnote{We obtain this bipartite graph via a random bipartition, which is described subsequently after the simple version.} with vertices being red, $R$, or blue, $B$, and $|R| = |B| = n$, and every vertex $r\in R$ has at least one neighbor in $B$ and we begin with no knowledge of the edges. The goal is to design a randomized algorithm (\Cref{alg SRCC}) which makes $O(n)$ queries in expectation, and learns at least one edge incident to every red vertex, thereby leading to a new graph with $\leq n$ connected components (down from the $2n$ singletons we began with). This forms the heart of the final maximal spanning forest algorithm, which is described in~\Cref{subsec:2}.

Before we describe the idea behind our algorithm, it is worthwhile describing the idea in~\citet{Apers} in unweighted graphs (who also solve the above problem) and why it fails to generalize with weights. 	
Their algorithm first queries the number of blue neighbors for every red vertex using $O(n)$ $\CUT$ queries; this {\em crucially} uses that the graph is unweighted. After this, 
they partition the vertices of $R$ into $\lceil \log_2 n \rceil$ classes where class $i$ contains vertices with {\em degree} $\approx 2^i$, and for
each of these classes, the algorithm samples a subset of vertices from $B$ with probability $\approx 1/2^i$. It's not hard to show that $\Theta(1)$ of the vertices in $R$ has 
	exactly one neighbor in the corresponding sampled subset in $B$; in other words, with high probability there is a matching of size $cn$ for some $c<1$. 
 At this point, \citet{Apers} uses an algorithm by~\citet{GrebinskiK00} which gives an $O(n)$ query algorithm to recover all the edges of this matching.

\paragraph{\em The Degree Issue with Weights and Our Bypass.} As noted in~\cite{ChakrabartyL23}, the main issue in implementing the above method in weighted graphs is in the first step of figuring out the number of blue neighbors (let's call this the blue-degree) of a red node. While this is a near triviality in unweighted graphs, with weighted graphs there are some provable hardness (see~\cite{ChakrabartyL23}). Indeed, the lion's share of the Monte-Carlo $O(n\log\log n)$ query algorithm in~\cite{ChakrabartyL23} is spent in {\em estimating} the blue-degrees of every red node. It is left as an open question whether this can be done in $O(n)$ time, and if so, their algorithm could perhaps be modified to give an $O(n)$ time algorithm. We do not resolve this ``degree-estimation'' question; indeed, very recently,~\cite{chakraborty2022support} (cf. Theorem 4.1) prove a super-constant hardness on the problem of estimating the degree and thus this route possibly cannot give a $O(n)$-algorithm.
Rather,
{\em our main insight is that the above idea of ``sampling inversely proportional to degree'' can be morally simulated even without knowing the degrees.}
This idea could potentially be useful in other applications.

\paragraph{\em Our Algorithm in a Nutshell.} 
We proceed in $\lceil \log_2 n \rceil$ iterations. 
We maintain a subset $\calR \subseteq R$ of red vertices for which we haven't found a blue neighbor, and initially $\calR = R$.
In the first iteration, we sample a subset $\calB \subseteq B$ where every vertex is present with probability $1/n$. Then, we use known graph reconstruction algorithms (\Cref{lemma eobm} or~\Cref{lemma eochoi}) to learn the edges in the subgraph $E(\calR,\calB)$. Since $|\calB|$ is small, this is a sparse subgraph, and the number of queries needed is small. Next, we {\em remove} every vertex in $\calR$ for which we have found an edge and proceed to the next iteration.
Now we sample $\calB \subseteq B$ with probability $2/n$, and repeat the same procedure, always removing vertices from $\calR$. In the $i$th iteration, the sampling probability is $2^i/n$, and therefore in the $\log n$th iteration, we reconstruct the graph $E(\calR,B)$ where recall $\calR$ is the subset of vertices of the original $R$ for whom we haven't discovered an edge. And so, by the end of these iterations, we would've learnt at least one neighbor for every vertex in $R$ completing what we set out to do.

Why is the query complexity of the above algorithm small? Indeed, the worry is that when $\calB$ is as big as $\Omega(n)$, the graph between $\calR$ and $\calB$ may no longer be sparse.
We prove (see~\Cref{corollary alg1property}) that this cannot be the case by noting that by the time $\calB$ is ``large'', all the high-degree vertices in the original $\calR$ would already have been removed. 
In particular, if the degree of a red vertex $r$ is $d(r)$, which remember is something that we don't know, then this would have been removed by the $\approx \log_2 (n/d(r))$th iteration.
For instance, if $|\calB| = \Theta(n)$, then $\Theta(\log n)$ rounds of the process must have passed, and it's highly likely that the only vertices remaining in $\calR$ would have degree $\Theta(1)$. And so, in this iteration the graph is sparse as well. This explains the key new idea behind our randomized algorithm and is formally proven in~\Cref{lemma alg1property}. 

\paragraph{Deterministic Algorithm (\Cref{thm:2}) Idea.}
Our deterministic algorithm is actually pretty simple and stems from the observation that with $\CUT$ queries one can learn neighborhoods of ``high'' degree vertices paying $\ll \log n$ queries per vertex.
The algorithm keeps doing this and growing connected components BFS style till the number if edges across the components become much smaller than $O(n\log n)$. Since the graph is unweighted, this estimate can be maintained. Once the graph becomes this sparse, once again exploiting the power of $\CUT$ queries, the whole graph can be reconstructed using known results. The latter step requires some non-trivial work since naively we only obtain the information which pairs of components have an edge between them, which is enough for answering the question whether a graph is connected or not; finding the true edges requires the full power of the $\CUT$-queries.
Balancing these two ideas gives an $O\left(\frac{n\log n}{\log\log n}\right)$ query algorithm. We don't believe this is the correct answer, but it is perhaps a first step in obtaining $\ll n\log n$ query algorithms a la~\citet{Harvey08}, which, recall, would also work with even a weaker ``OR'' query model and for which the query complexity is optimal. 

\subsection{Related Works}\label{sec:rel}

The question of reconstructing the whole graph using $\CUT$ queries has almost been resolved after a long series of works~(cf. \citet{GrebinskiK00,AlonBKRS02,AlonA05,ReyzinS07,ChoiK10,Mazzawi10,BshoutyM11,BshoutyM12,Choi13}).
For unweighted undirected graphs with $m$ edges, \citet{ChoiK10} proved the existence of non-adaptive deterministic algorithms making $O(\frac{m\log\frac{m}{n}}{\log m})$ queries, and~\cite{Mazzawi10}
described an efficient adaptive deterministic algorithm with similar query complexity. This query complexity is information theoretically optimal.
For (non-negative) weighted graphs,~\citet{Bshouty2011parity} prove the existence of non-adaptive $O(\frac{m\log n}{\log m})$-query algorithms. Then in~\citet{BshoutyM11} the authors describe a deterministic $O(\frac{m\log n}{\log m} + m\log\log m)$-query adaptive algorithm. If we allow randomization, then~\citet{Choi13} gives an efficient, randomized $O(\frac{m\log n}{\log m})$ query algorithm which is information theoretically tight as well. For our efficient randomized algorithm, we use this result. It is an open question to design efficient non-adaptive algorithms to reconstruct using optimal query complexity.

The question of studying just whether an undirected graph is connected using $\CUT$ queries was perhaps first explicitly noted in~\citet{Harvey08} due to connections to submodular function minimization.
The same question for minimum cuts was initiated by the paper~\citet{RubinsteinSW18} who described
an $O(n \mathrm{polylog}~n)$ algorithm to find the minimum cut. In fact, the paper of~\citet{Apers} mentioned above, gives a randomized Monte Carlo $O(n)$-query algorithm to solve this problem.
In unweighted graphs, it is not known whether this query complexity is tight, and the only lower bound is an $\Omega(n\log\log n/\log n)$-lower bound that follows from the communication complexity of connectivity result of~\cite{RS95}. 
For weighted graphs, the best known algorithm to solve the minimum-cut is a randomized Monte Carlo $O(n\mathrm{polylog}~n)$-query algorithm by~\citet{MN20}. There is no $O(n)$ algorithm known for finding the minimum cut in the weighted case.

Graph reconstruction has also been studied under other query models. In particular, there is the weak ``OR'' query model (also known as independent set (IS) queries), where one asks for a subset $S$ and obtains only the information whether $S$ is an independent set or not. A slightly different model called the bipartite independent set (BIS) query passes two disjoint subsets $A, B$ and obtains whether there exists an edge with one endpoint in $A$ and one endpoint in $B$. 
This model is more related to ``group testing'' question~\citet{Dorfman43} where one only gets a weak signal from a query. \citet{AngluinC08} described an $O(m\log n)$-query algorithm to learn the whole graph using IS queries, and this is information theoretically tight. The question of reconstructing the spanning tree from such weaker models have also been studied. It is not hard to see an $\Omega(n^2)$ lower bound\footnote{Consider a graph on $2n$ vertices where we have two cliques on $n$ vertices connected with a single edge. Any IS query containing more than one vertex from any of the parts gives no information. So, the problem reduces to finding 
a $1$ in an $n^2$-dimensional vector where we can only query singletons.} for IS queries. Furthermore, the algorithm in~\citet{Harvey08} works with BIS queries, and indeed, as mentioned above, 
in this weaker model this factor is tight. 
In~\citet{assadi}, the authors study the rounds-vs-query trade-off for the problem of learning a spanning tree for BIS queries, and give a near sharp resolution. 
Finally, other properties such as estimating the number of edges (\citet{BeameHRRS18,ChenLW20,AddankiMM22}, etc) have been studied under these models, and more generally this area of understanding the query complexity of learning/estimating graph parameters under a variety of models is a rich and relevant area of study.

%% file: section2_prelim.tex
\section{Preliminaries and Subroutines} 
We begin by defining two other notions of queries which have been used in the literature and is often more convenient to use.
We begin with the notion of $\CROSS$ queries. Given a graph $G = (V, E, w)$ and two disjoint subsets $A, B \subseteq V$, $\CROSS(A, B)$ returns $\sum_{e\in E(A, B)}w(e)$, where $E(A, B)$ is the set of vertex pairs $(a, b) \in A \times B$ such that $(a, b) \in E$. It's easy to see the following.

\begin{proposition}\label{prop:1}
	A $\CROSS$ query can be simulated by 3 $\CUT$ queries.
\end{proposition}
\begin{proof}
	$\CROSS(A,B) = \frac{1}{2}\cdot \left(\CUT(A) + \CUT(B) - \CUT(A\cup B)\right)$.
\end{proof}
\def\ADD{\mathsf{ADDITIVE}}

Next is the notion of an $\ADD$ query. Given a graph $G = (V, E, w)$ and subset $S\subseteq V$, the additive query $\ADD(S)$ returns the sum of weights
$\sum_{e = (a,b)\in E, a\in S, b\in S} w(e)$.

\begin{proposition}\label{prop:2}
	Any algorithm making $t$ $\ADD$ queries can be simulated by an algorithm making $n + t$ $\CUT$ queries.
\end{proposition}
\begin{proof}
	This follows by noting that $\ADD(S) = \frac{1}{2}\cdot \left(\sum_{v\in V} \CUT(v) - \CUT(S) \right)$.
	That is, once $\CUT(\{v\})$ is known for all the $n$ different $v\in V$, any $\ADD$ query can be answered using a single $\CUT$ query.
\end{proof}

We then state the graph reconstruction results with optimal query complexity and the minor modifications we need to make for our purposes. 
The first result is a deterministic, non-adaptive query algorithm by~\cite{Bshouty2011parity} with optimal query complexity.

\begin{lemma} \label{lemma bm}~\citep[][Corollary 4]{Bshouty2011parity}.
There exists a non-adaptive determinstic algorithm $\mathsf{GraphReconstructionBM}$ that uses $O(\frac{m \log n}{\log m})$ $\ADD$ queries and reconstruct any weighted hidden graph with at most $m$ edges.
\end{lemma}
\noindent
Note that the above algorithm uses additive queries. An additive query on $S \subseteq V$ returns the sum of edge weights in the induced subgraph with vertex set $S$, and in general is 
stronger than $\CROSS$ queries\footnote{While $3$ additive queries can simulate a $\CROSS$ query, it was shown in~\cite{LeeMS21} that $\Omega(n)$ $\CROSS$ queries may be needed to simulate an additive query.}.
Observe that if $G = (R, B)$ is a bipartite graph, an additive query on $S$ is equivalent to a $\CROSS$ query $\CROSS(S \cap R, S \cap B)$. We will only be reconstructing bipartite graphs, and therefore, we only need the following corollary which is a restatement of the above lemma.

\begin{corollary} \label{lemma bmcross}
There exists a non-adaptive, deterministic algorithm $\mathsf{GraphReconstructionBM}$ ($\mathsf{GRBM}$ in short) that uses $O(\frac{m \log n}{\log m})$ $\CROSS$ queries and reconstruct any weighted hidden \textbf{bipartite} graph with at most $m$ edges.
\end{corollary}
\noindent
\cite{Bshouty2011parity} only proves the existence of a query algorithm but neither gives an explicit construction, nor does it give a polynomial time recovery algorithm. Indeed, it is an outstanding open question to obtain a non-adaptive and/or deterministic algorithm with this query complexity. The polynomial time complexity was resolved in~\cite{Choi13} via an adaptive, randomized Monte Carlo\footnote{The error probability in Theorem 1 of~\cite{Choi13} is only stated as $1-o(1)$, however, in the statement of Theorem 3 which uses Theorem 1, it is mentioned as $O(\log m/m)$.} algorithm.

\begin{lemma} \label{lemma choi}~\citep[Paraphrasing][Theorem 1]{Choi13}.
One can construct a randomized polynomial time adaptive algorithm $\mathsf{GraphReconstructionChoi}$ ($\mathsf{GRC}$ in short) which given an unknown weighted graph $G$ on $n$ vertices with at most $m$ edges, makes $O\left(\frac{m\log n}{\log{m}}\right)$ $\CROSS$ queries, and with probability $1 - O(\frac{\log m}{m})$ returns the edges of $G$ and their weights.
\end{lemma}

As stated, both results above assume knowledge of this parameter $m$ which is an upper bound on the number of edges. 
For our purposes, the number of edges would be unknown (in fact, random variables), and so we use a simple modification via a doubling-trick 
to get an efficient algorithm that recovers the hidden edges of a graph using $\CROSS$ queries that doesn't require to know the number of the hidden edges to be known.
We state these as theorems below.

\begin{subtheorem}{theorem}
\begin{theorem}
~\citep[Follows from][Corollary 4]{Bshouty2011parity}. \label{lemma eobm}
There exists an adaptive, deterministic algorithm $\GRone$ which takes input a bipartite graph $G = (U, V)$ on $n$ vertices and $m$ edges such that $U$ and $V$ are known but $m$ is unknown to the 
algorithm, and reconstructs the edges of $G$ along with their weights, making $\frac{C_{\GRone}m\log n}{\log m}$ $\CROSS$ queries, for some absolute constant $C_{\GRone} > 0$.
\end{theorem}

\begin{theorem}~\citep[Follows from][Theorem 1]{Choi13}. \label{lemma eochoi}
	There exists an adaptive, randomized algorithm $\GRtwo$ which takes input a bipartite graph $G = (U, V)$ on $n$ vertices and $m$ edges such that $U$ and $V$ are known but $m$ is unknown to the 
	algorithm, and either reconstructs the edges of $G$ along with their weights, or aborts. The probability that the algorithm aborts or makes more than $\frac{C_{\GRtwo}m\log n}{\log m}$ $\CROSS$ queries, for some absolute constant $C_{\GRtwo} > 0$, is at most $O(\frac{\log m}{m})$.
\end{theorem}
\end{subtheorem}


\begin{proof}
    See~\Cref{Appendix:A}.
\end{proof}

The next subroutine we need is the simple DFS-style algorithm to find a spanning forest. In particular, if the number of connected components is $q$, then the remaining edges can be found in $O(q\log n)$ many queries. This was also used by~\cite{Apers} for unweighted graphs but an inspection of their proof shows that it readily works with weights as well. The idea is that even in a weighted non-negative  vector, a single element in the support can be found using binary search, and every connected component can find an edge coming out of it (if they exist) in $O(\log n)$ queries.

\begin{lemma} \label{lemma dfs}~\citep[Paraphrasing][Lemma 5.1]{Apers} Let $G = (V, E)$ be an $n$-vertex weighted graph with non-negative weights. Let $G'$ be a contraction of $G$ with $q$ many supervertices, which are given explicitly as the partition $P = {A_1, \cdots, A_q}$ of $V$. There is a deterministic algorithm $\mathsf{DFSSpanningForest}$ that takes in $G, G'$ and outputs a set of edges $F \subseteq E$ that form a spanning forest of $G'$ and makes $O(q \log n)$ $\CUT$ queries to $G$.
\end{lemma}

For our deterministic algorithm, we need the following result about the coin-weighing problem.
Let $x\in \{0,1\}^N$ be an unknown Boolean vector, and suppose we can only access it via querying a subset $S\subseteq [N]$
and obtaining $\sum_{i\in S} x_i$. 

\begin{restatable}[\cite{Bshouty09}]{lemma}{bshouty}
\label{lem:bshouty}
	Let $x\in \{0,1\}^N$ be an unknown Boolean vector accessed via querying a subset $S\subseteq [N]$ and obtaining $\sum_{i\in S} x_i$ 
 (sum-query access). If $x$ has $d$ ones, then there is a polynomial time deterministic algorithm $\Bshouty$ to reconstruct $x$ which makes $O(d\log(N/d)/\log d)$ sum-queries.
\end{restatable}

%% file: section3_main_alg.tex
\section{Zero Error Randomized $O(n)$-Query Algorithm}\label{sec:3}

As mentioned in~\Cref{sec:tech}, to underscore our main technical contribution, in~\Cref{subsec:1} we focus on the problem of given a bipartite weighted graph $G = (R \sqcup B, E)$ with $2n$ singletons with the promise that every vertex in $R$ has a neighbor in $B$, how to make $O(n)$ queries in expectation to learn a subgraph where we learn at least one neighbor incident to every vertex in $R$. Thus, the number of connected components go down from $2n$ to $< n$. 
This algorithm is described as~\Cref{alg SRCC}. 

The generalization of this that is the subroutine to the spanning forest algorithm is~\Cref{alg RCC}. This algorithm takes a graph with $t$ connected components and makes $O(t)$ queries in expectation to learn edges which leads to $<ct$ connected components for some $c < 1$. This then can be plugged into the~\cite{Apers} framework to give the spanning forest algorithm. This is described in~\Cref{subsec:2}.

%



\subsection{The Main Idea: Learning edges in a bipartite graph}\label{subsec:1}

We assume our unknown, weighted, undirected graph $G = (R\sqcup B, E)$ is {\em bipartite} with vertices being red, $R$, or blue, $B$, and $|R| = |B| = n$, and every vertex $r\in R$ has at least one neighbor in $B$. We also know the bipartition.
The goal is to learn at least one edge incident on every red vertex. In this section, we describe a randomized algorithm to do this making $O(n)$ queries in expectation. As explained in~\Cref{sec:tech}, this forms the main heart of the spanning forest algorithm since these learnt edges reduces the number of connected components by a constant factor.

The algorithm proceeds in $\lceil \log_2 n\rceil$ iterations. It maintains a subset $\calR_i \subseteq R$ with $\calR_0 = R$ which is supposed to signify the subset of red vertices
which hasn't discovered an edge incident on them.
In the $i$th iteration, the algorithm samples a subset blue vertices $\calB_i \subseteq B$ at a rate $\frac{2^i}{n}$. 
We then reconstruct all the edges in $E(\calR_i, \calB_i)$.
Then, we remove any vertices $r$ in $\mathcal{R}_i$ participating in any of these reconstructed edges since we have already found an edge incident on them, to get $\calR_{i+1}$.
Note that in the final round $\ell := \lceil \log_2 n\rceil$, the subset $\calB_\ell = B$, and thus we would reconstruct all the edges between $E(\calR_\ell, B)$. In particular, since
all vertices in $\calR_\ell$ have a neighbor in $B$, we would suceed with probability $1$. 

\begin{algorithm}
\caption{$\mathsf{SkeletonReduceConnectedComponents}$}\label{alg SRCC}
\begin{algorithmic}[1]
\Require $\CUT$ access to a bipartite graph $G = (R\sqcup B,E)$ with $|R|=|B|=n$ and unknown edges with weights. Assumption:  each vertex in $R$ each has a neighbor in $B$. 
\Ensure Set of edges with at least one edge incident on every $r\in R$.
\State $\mathcal{R}_0 \gets R$.
\State $i \gets 0$.
\For{$i \le \lceil\log_2 n\rceil$}
    \State Sample a set $B_i$ with every $b \in B$ sampled with probability $\frac{2^i}{n}$.
    \State Recover edges $E_i := E(\mathcal{R}_i, B_i)$ using $\mathsf{GR1}$ in~\Cref{lemma eobm}.
    \State $\mathcal{R}'_i = \{v \in \mathcal{R}_i |\ v \text{ is an endpoint of some } e \in E_i\}$.
    \State $\mathcal{R}_{i + 1} \gets \mathcal{R}_i - \mathcal{R}'_i$.
    \State $i \gets i + 1$.
\EndFor
\State \Return $\bigsqcup_{i \in [\lceil\log n\rceil]} E_i$.
\end{algorithmic}
\end{algorithm}

\begin{theorem} \label{corollary alg1property}
The expected number of edges is $\ex{|\bigsqcup_{i \in [\lceil\log n\rceil]} E_i|} \leq 5n$.
\end{theorem}

\noindent
{\bf Remark.} \emph{
Since we run for $\Theta(\log n)$ phases, it can be shown that if 
$|\bigsqcup_{i \in [\lceil\log n\rceil]} E_i| = \Theta(n)$, then the number of queries needed to recover using $\GRone$ is also $\Theta(n)$. This uses a simple analytic claim 
present in~\Cref{claim 3}. This can be used to show that the expected running time of the above algorithm is $\Theta(n)$.
}

Define $d(r): R \to [n]$ to be a function from $r \in R$ to the number of $r$'s neighbors in $B$. Note that we don't know how to obtain or even estimate $d(r)$ in $O(1)$ $\CUT$ queries, and this
definition is {\em only for analysis}. The proof of the above theorem follows almost immediately from the next lemma: for each edge $e\in E$ let $X_e$ be the indicator that it is present in the union.
Then the next lemma implies that $\ex{|\bigsqcup_{i \in [\lceil\log n\rceil]} E_i|} \leq  \sum_{r\in R} \sum_{b: (r,b)\in E} \frac{5}{d(r)} = 5n$.

\begin{lemma} \label{lemma alg1property}
Fix an edge $e = (r, b) \in E$ where $r \in R, b \in B$. $\pr{e \in \bigsqcup_{i \in [\lceil\log n\rceil]} E_i} \le \frac{5}{d(r)}$.
\end{lemma}

\begin{proof}
Let $\mathcal{E}_i$ denote the event that $e \in E_i$, that is, the reconstructed graph in the $i$th round. 
Our goal is to show $\pr{\bigsqcup_{i}\mathcal{E}_i} \le \frac{5}{d(r)}$. 

For the event $\mathcal{E}_i$ to occur, the vertex $r$ must have ``survived'' the first $(i-1)$ rounds (that is, no edge incident to $r$ was discovered before) and the vertex $b$ must be in the set $\calB_i$.
To this end, for any $0\leq j\leq i-1$, let $\mathcal{F}_j$ be the event that no neighbor of $r$ is in $\calB_j$. First note that 
$$\pr{\mathcal{F}_j} = \left(1 - \frac{2^j}{n}\right)^{d(r)}$$
since none of its $d(r)$ neighbors are present in $\calB_j$. Since all the $\mathcal{F}_j$'s, $0\leq j\leq i-1$, and the event $b\in \calB_i$ are mutually independent, 
we get that
\begin{equation}\label{eq:aha}
\pr{\mathcal{E}_i} = \pr{b \in \calB_i}\prod_{j = 0}^{i - 1} \pr{\mathcal{F}_j} = \frac{2^i}{n}\cdot \prod_{j = 0}^{i - 1}\left(1 - \frac{2^j}{n}\right)^{d(r)} 
< \frac{2^i}{n}\cdot e^{-\frac{(2^i - 1)d(r)}{n}}
\end{equation}
where we used $\forall x\neq 0, 1+x < e^x$ and geometric series formula for the inequality.

Note that when $2^i \approx n/d(r)$, the RHS is $\approx 1/d(r)$, and for every other $i$ either the first term in the product or the second term in the product in the RHS are orders of magnitude smaller.
Therefore, the sum of $\pr{\mathcal{E}_i}$'s as $i$ ranges from $0$ to $\log n$ can be bounded by $O(1/d(r))$. More precisely, 
let $i^* = \lfloor \log_2 \frac{n}{d(r)} \rfloor$. We consider the following two cases:

\begin{itemize}
    \item $i \le i^*$: since $\frac{(2^i - 1)d(r)}{n} \ge 0$, $e^{-\frac{(2^i - 1)d(r)}{n}} \le 1$. In the summation $\sum_{i = 0}^{i^*}\frac{2^i}{n}e^{-\frac{(2^i - 1)d(r)}{n}}$, each term is at most $\frac{2^i}{n}$, thus 
    \begin{equation}\label{eq:c1}
    \sum_{i = 0}^{i^*}\frac{2^i}{n}e^{-\frac{(2^{i} - 1)d(r)}{n}} \le \sum_{i = 0}^{i^*} \frac{2^i}{n}\le 2 \cdot \frac{2^{i^*}}{n} \le \frac{2n}{d(r)} \frac{1}{n} = \frac{2}{d(r)}.
    \end{equation}
    \item $i \ge i^* + 1$: let $k = i - (i^* + 1)$. Consider the ratio between $i + 1$th term and $i$th term, 
    \begin{alignat}{7}
    	  \left(\frac{2^{i + 1}}{n}e^{-\frac{(2^{i + 1} - 1)d(r)}{n}}\right) / \left(\frac{2^{i}}{n}e^{-\frac{(2^{i} - 1)d(r)}{n}}\right)  &&~=~& 2e^{-\frac{(2^{i + 1} - 1)d(r)}{n} + \frac{(2^{i} - 1)d(r)}{n}} ~=~2e^{\frac{2^id(r) - 2^{i+1}d(r)}{n}} \notag \\
    	  &&~=~& 2e^{-\frac{2^id(r)}{n}} 
    	  ~\leq~ 2e^{-\frac{2^{k + \log_2\frac{n}{d(r)}}d(r)}{n}}  \label{eq:007}\\
    	  &&~=~& 2e^{-2^k} \notag
    \end{alignat}
where \eqref{eq:007} follows because $i = k+ (i^* + 1) \ge k + \log_2\frac{n}{d(r)}$.
Notice that $2e^{-2} < 0.3$ and $2e^{-2} \cdot 2e^{-4} = 0.001$. The numerical value imply that most of the mass in $\sum_{i = i^* + 1}^{\lceil \log n\rceil}\frac{2^{i}}{n}e^{-\frac{(2^{i} - 1)d(r)}{n}}$ is concentrated in the first two terms. One can easily check that $\sum_{i = i^* + 1}^{\lceil \log n\rceil}\frac{2^{i}}{n}e^{-\frac{(2^{i} - 1)d(r)}{n}} \le \frac{3}{2} \cdot \frac{2^{i^* + 1}}{n}e^{-\frac{(2^{i^* + 1} - 1)d(r)}{n}}$, in other words, the sum is at most $\frac{3}{2}$ of the first summand. 

The first term is $\frac{2^{i^* + 1}}{n}e^{-\frac{(2^{i^* + 1} - 1)d(r)}{n}} \le \frac{2}{d(r)}e^{-\frac{(2^{i^* + 1} - 1)d(r)}{n}} \le \frac{2}{d(r)}$. Therefore
\begin{equation}\label{eq:c2}
\sum_{i = i^* + 1}^{\lceil \log n\rceil}\frac{2^{i}}{n}e^{-\frac{(2^{i} - 1)d(r)}{n}} \le \frac{3}{d(r)}.
\end{equation}
\end{itemize}
\Cref{eq:c1,eq:c2} together imply  $\sum_{i = 0}^{\lceil\log n\rceil} \pr{\mathcal{E}_i} \leq 5/d(r)$ by union bound.
\end{proof}

\noindent
\textbf{Remark}: \emph{It is worthwhile to note that the above algorithm {\em doesn't} use {\em non-negativity} of weights. This is because of the promise that every red node has at least one blue neighbor.
In the full algorithm (see~\Cref{subsec:2}, or~\cite{Apers,ChakrabartyL23}), the bipartition is formed randomly, and we abort the bipartition if less than a constant fraction 
of red nodes have blue neighbors. However, checking whether a red node $r$ has a blue neighbor can be done with one $\CROSS$ query between $r$ and $B$ {\em only if} the weights are non-negative.
If weights were allowed to be negative, this step wouldn't work. Indeed, with negative weights, it is not to hard to show that $\widetilde{\Omega}(n^2)$ queries would be needed to solve the connectivity question.
}

\subsection{Full Algorithm for Learning Spanning Forest in $O(n)$ queries}\label{subsec:2}

{\em Outline.}
Like previous algorithms~\cite{Apers, ChakrabartyL23}, our spanning forest algorithm proceeds in a Bor\r{u}vka style framework.
It proceeds in phases, and in each phase, the algorithm
	begins with $t$ of connected components $C_1, \ldots, C_t$ and a set of learnt edges such that
 each $C_i$ is connected in these learnt edges. In a phase, the algorithm discovers edges crossing these components (thereby connecting them) in such a way that 
	$O(t)$ $\CUT$ queries are made in expectaion, and with the newly discovered edges the number of connected components goes down from $t$ to $ct$ where $c$ is a constant $< 1$. 
This phase forms the heart of the matter, since if this can be done, then simply repeating this till the number components becomes $1$ leads to the $O(n)$ algorithm. 
The details of this is given in~\Cref{alg SFA}, and we now proceed to understanding individual phases. Before doing so, let us introduce the definition of representatives of connected components a la ~\cite{Apers}.

{\em Representatives and Active/Inactive Vertices.} When the connected components are no longer singletons, for each connected component $C_i$ we need to select a vertex as the {\em representative} for $C_i$. 
To do so, one defines 
``active" and ``inactive" vertices. For a fixed $C_i$, if there is an edge $e = (u, v)$ with $u \in C_i$ and $v \in C_j$ for any $j \neq i$, in other words, $u$ has at least a neighbor in a different connected component, we say that the vertex $u$ is ``active". Otherwise $u$ is ``inactive". For a connected component $C_i$, any ``active" vertex in $C_i$ can be its representative. After a phase is run, a representative $u$ can become ``inactive". In this case, we can go through vertices in the (expanded) connected components that $u$ belongs to, that are not ``inactive", and find an active one as the new representative. It can also happen that a connected component has multiple representatives after a phase, in which case we assign an arbitrary one as the representative. Since an ``inactive" vertex will remain ``inactive" throughout, one can maintain a representative for each connected components with an $O(n)$ overhead on the queries. Detailed argument can be bound in the proof of~\Cref{theorem 15}. Henceforth, we assume that each component $C_i$ has a {\em representative} vertex 
$c_i \in C_i$ with the guarantee that $c_i$ has an edge to $V\setminus C_i$. If not, we learned a connected components in the final maximal spanning forest and we don't need to pass it in as part of the input. 

{\em Phase of Reducing Connected Components.}
The following Monte Carlo algorithm (\Cref{alg RCC}), named $\mathsf{ReduceConnectedComponents}$ is the general version of $\mathsf{SkeletonReduceConnectedComponents}$ (\Cref{alg SRCC}). It describes one phase of the algorithm, where we start with $t$ connected components and end with at most $\frac{7}{8}$ connected components with constant probability. 

Lines~\ref{RCC color start} to~\ref{RCC color end} of \Cref{alg RCC} first creates the bipartite graph that was assumed in \Cref{alg SRCC}  by coloring each connected component either red or blue
uniformly at random. It is designed so as to ensure that a constant fraction of the ``red representatives'' have an edge to some blue-component vertex. Recall in~\Cref{alg SRCC}, we assumed {\em every} red vertex had a blue neighbor, but even if a constant fraction had them, even then there would be a constant factor drop. The rest of the algorithm is very similar to~\Cref{alg SRCC} with 
a couple of differences: (i) we maintain a ``budget'' on the queries and if we ever cross it we ABORT, and thus get a Monte-Carlo algorithm, and (ii) instead of the non-constructive deterministic $\GRone$, we replace it with the randomized Monte Carlo $\GRtwo$ from~\Cref{lemma eochoi}, which always ensuring we don't ever make $\omega(t)$ queries. 


\begin{algorithm}[ht!]
\caption{$\mathsf{ReduceConnectedComponents}$}\label{alg RCC}
\begin{algorithmic}[1]
\Require $G_2$ with $t$ connected components $C_1, \cdots, C_t$, each with $\ge 1$ representative.
\Ensure Either return a graph with $\le \frac{7}{8}t$ connected components with constant probability or $\mathsf{ABORT}$.
\State $R \gets \emptyset$.
\State For $i \in [t]$, pick color $c_i$ from $\{red, blue\}$ uniformly at random. Color all vertices in $C_i$ to $c_i$. \label{RCC color start}
\State $B \gets $ vertices colored blue.
\State $R \gets $ red representatives with at least one blue neighbor. \Comment{Can be found with $O(t)$ $\CROSS$ queries.}
\If{$|R| < \frac{t}{8}$}
$\mathsf{ABORT}$. \label{RCC color end}
\EndIf
\State $i \gets 0$.
\State $\budget \gets 120\max(C_{\GRone}, C_{\GRtwo})t.$ \Comment{Recall $C_{\GRone}$ and $C_{\GRtwo}$ from~\Cref{lemma eobm}}
\State $\mathcal{R}_0 \gets R$.
\For{$i \le \log t$}
\State Sample a set $B_i$ with every $b \in B$ sampled w.p. $\frac{2^i}{t}$.
    \While{$\budget > 0$}
	\State Run $\mathsf{GR1}$ or $\mathsf{GR2}$ as described in~\Cref{lemma eobm} or~\Cref{lemma eochoi} on $E_i := E(\mathcal{R}_i, B_i)$, always decrementing $\budget$ by $1$ whenever a $\CROSS$ query is made.
    \If{$\budget = 0$}
        $\mathsf{ABORT}$.\Comment{$\mathsf{ABORT}$ either because sampled edges $> 5t$ or running $\mathsf{GR2}$ has failed}\label{ABORT budget}
    \EndIf
    \State $\mathcal{R}'_i = \{v \in \mathcal{R}_i |\ v \text{ is an endpoint of some } e \in E_i\}$.
    \State $\mathcal{R}_{i + 1} \gets \mathcal{R}_i - \mathcal{R}'_i$.
   \State $i \gets i + 1$.
   \EndWhile
\EndFor
\State \Return{$G_2 + \bigsqcup_{i \in [\lceil\log t\rceil]} E_i$.}
\end{algorithmic}
\end{algorithm}

\begin{lemma} \label{lemma alg2property}
	Let $G = (V, E)$ be a weighted graph with non-negative weights. Suppose $G_2$ is a subgraph of $G$ with $t$ connected components for some $t \ge \frac{n}{\log n}$.~\Cref{alg RCC} makes $O(t)$ queries, and with probability $\geq 1/10$ returns a graph with $\le \frac{7}{8}t$ connected components, or returns $\mathsf{ABORT}$. 
\end{lemma}

%

\begin{proof}
	Let $C_1, \cdots, C_t$ be the connected components of $G_2$.~\Cref{alg RCC} begins by coloring the components red or blue uniformly at random. If a component is colored red/blue, we color all its vertices with the same color. Let $B$ be the set of vertices colored blue. For a red vertex $r$, we use $d(r)$ to denote the number of blue neighbors; recall, a neighbor is a vertex $b$ such that the weight $w(r,b) > 0$.
	Call a representative $r$ of a component ``good" if it is colored red and $d(r) \ge 1$. Note that a representative is good with probability $\geq 1/4$; its component has to be colored red and one of its
	neighbor's component has to be colored blue. Since it has at least one neighbor outside its component (since it's a representative), the probability follows. 
	Hence $\ex{|R|} \geq t/4$ and $\pr{|R| \le \frac{t}{8}} \le \frac{1}{2}$. Therefore, the probability~\Cref{alg RCC} aborts in~\Cref{RCC color end} is at most a half.
	
	Fix an iteration $i$.
	Let $X_i$ be the number of edges in $E(\mathcal{R}_i, B_i)$ which is a random variable. The induced subgraph $(\mathcal{R}_i, B_i)$ has at most $n$ vertices. 
	By~\Cref{lemma eobm} or~\Cref{lemma eochoi}, it takes at most $$\max(C_{\GRone}, C_{\GRtwo}) \cdot \left(\frac{X_i \log n}{\log X_i}\right)$$ queries to recover the edges in $E(\mathcal{R}_i, B_i)$. More precisely, if we run the deterministic algorithm from~\Cref{lemma eobm}, we recover these edges with probability $1$. If we run the efficient randomized algorithm from~\Cref{lemma eochoi}, we recover these edges with probability $1-O(\log m/m)$.
	
	By~\Cref{lemma alg1property}, $\ex{\sum_{0\le i \le \lceil\log t\rceil}X_i} = \sum_{e}\pr{e \in \cup_{i \in [\log t]} E_i} \le \sum_{r \in \mathcal{R}_0}\frac{5}{d(r)} \cdot d(r) = 5|\mathcal{R}_0| \le 5t$. By Markov's Inequality, $\pr{\sum_{0\le i \le \lceil\log t\rceil}X_i \ge 15t} \le \frac{1}{3}$. Conditioned on the event $\sum_{0\le i \le \lceil\log t\rceil}X_i \le 15t$, it takes
	\begin{align*}
		\sum_{i = 0}^{\lceil\log t\rceil} \max(C_{\GRone}, C_{\GRtwo}) \left(\frac{X_i \log n}{\log X_i}\right)
		&\le \max(C_{\GRone}, C_{\GRtwo}) \log n \sum_{i = 0}^{\lceil\log t\rceil}  \left(\frac{X_i}{\log X_i}\right)
		\\ &\le \max(C_{\GRone}, C_{\GRtwo}) \log n (4 \cdot \frac{15t}{\log t})\quad (\text{by~\Cref{claim 3}})
		\\ &\le 120 \max(C_{\GRone}, C_{\GRtwo}) t
	\end{align*}
	many queries to recover all edges with probability $\geq 1 - O(\log t\log m/m)$. Note we have multiplied by $\log t$ to account for the union bound over the $\log t$ different reconstructions.
	
	Since the $\budget$ is initialized to be $120\max(C_{\GRone}, C_{\GRtwo})t$, we can recover all the edges when $\sum_{i}X_i \le 15t$ with failure probability $O(\frac{\log t\log m}{m}) \le \frac{1}{100}$,
	for large enough $m$. Contrapositively, if we abort in~\Cref{ABORT budget}, either $\sum_i X_i > 15t$ (which is $\leq 1/3$) or due to the failure probabilities of the $\log t$ different invocations of $\mathsf{GR2}$ (which is $\leq 1/100$). 
	The overall $\mathsf{ABORT}$ probability can be union bounded by 
	\begin{align*}
		&\pr{\mathsf{ABORT} \text{ in~\Cref{RCC color end}}} + \pr{\mathsf{ABORT} \text{ due to} \sum_i X_i 
			\geq 15t} + \\ & \pr{\mathsf{ABORT} \text{ due to failures of } \mathsf{GR2}} \le \frac{1}{2} + \frac{1}{3} + \frac{1}{100} < \frac{9}{10}.
	\end{align*}
	
	If~\Cref{alg RCC} doesn't abort, $\mathcal{R}_i = \emptyset$ at the end since every representative in $R$ are connected to some component in $B$ after edge sampling. Since $|R| \ge \frac{t}{8}$, the number of connected components shrinks by $\geq t/8$, and therefore $G_2 + \bigsqcup_{i \in [\lceil\log t\rceil]} E_i$ has $\le \frac{7}{8}t$ connected components. 
	
	We still need to show the number of queries is $O(t)$. We can check if a representative has a blue neighbor using $1$ $\CROSS$ query and there can be at most $t$ representatives. In the sampling/recovery stage, the total number of queries is bounded by the initialized value of $\budget$ which is $120\max(C_{\GRone}, C_{\GRtwo})t$. Hence the algorithm uses at most $O(t)$ queries. 
\end{proof}
One can now easily obtain a zero-error algorithm by repeating if an ABORT is encountered.

\begin{algorithm}[ht!]
\caption{$\mathsf{ZeroErrorReduceCC}$}\label{alg ZERCC}
\begin{algorithmic}
\Require $G_2$ with $t$ connected components $C_1, \cdots, C_t$, each with $\ge 1$ representative.
\Ensure A graph with $\le \frac{7}{8}t$ connected components.
\While{True}
  \State $Result \gets \mathsf{ReduceConnectedComponents}(G_2)$ (\Cref{alg RCC}).
    \If{$Result$ is not $\mathsf{ABORT}$}
       \State \Return{$Result$}
    \EndIf
\EndWhile
\end{algorithmic}
\end{algorithm}

\begin{corollary}\label{corollary alg3property}
	Let $G = (V, E)$ be a weighted graph with non-negative edge weights. Given a subgraph $G_2$ of $G$ with $t$ connected components for some $t \ge \frac{n}{\log n}$,~\Cref{alg ZERCC} returns a graph with $\le \frac{7}{8}t$ connected components making $O(t)$ $\CROSS$ queries in expectation.
\end{corollary} 

{\em Finishing up the Bor\r{u}vka style argument.} One can now use the above $\mathsf{ZeroErrorReduceCC}$ as a subroutine to get a zero-error randomized algorithm to find a spanning forest of an undirected weighted graph which makes $O(n)$ queries in expectation. We include the algorithm and analysis towards the final (maximal) spanning forest algorithm \footnote{The algorithm and analysis below are nearly identical to that in~\cite{ChakrabartyL23}}. 

\begin{algorithm}[ht!]
\caption{$\mathsf{SpanningForestAlgorithm}$}\label{alg SFA}
\begin{algorithmic}[1]
\Require $G = (V, E)$ with $w(e) \geq 0$ for all $e \in E$.
\Ensure A spanning forest of $G$.
\State $n \gets |V|$.
\State $i \gets 1$.
\State $G_i \gets (V, \emptyset)$.
    \For{$v \in V$} \label{Query1 start}
        \If{$|E(\{v\}, V\backslash v)| > 0$} 
            \State Mark $v$ as active.
        \Else \State Mark $v$ as inactive.
        \EndIf
    \EndFor \label{Query1 end}
\While{$i < 6\log \log n$}
\State   $G_{i + 1} \gets \mathsf{ZeroErrorReduceCC}(G_i)$ (\Cref{alg ZERCC}) \label{Query2}
\State   If any of old representatives in $G_{i}$ is still active, mark it as the next representative in $G_{i + 1}$. Otherwise keep marking vertices in $G_{i + 1}$ inactive until one finds an active vertex.  \label{Query1}
\State   $i \gets i + 1$
\EndWhile
\State $\mathcal{F} \gets$ Find the spanning forest of $G_i$ with $\mathsf{DFSSpanningForest}$ in Lemma~\ref{lemma dfs}. \label{Query3}
\State \Return{$\mathcal{F}$.}
\end{algorithmic}
\end{algorithm}

\begin{theorem} \label{theorem 15}
Let $G = (V, E, w)$ be a weighted graph in which we want to find a spanning forest where $|V| = n$. $V$ is known in advance but edges $E$ and edge weights $w \in (\mathbb{R}^+_0)^{|E|}$ are hidden.~\Cref{alg SFA} is a zero error algorithm that reconstructs a spanning forest of $G$ making $O(n)$ $\CUT$ queries in expectation. 
\end{theorem}

\begin{proof}
The queries we make in~\Cref{alg SFA} are from 3 parts: (1) calling~\Cref{alg ZERCC} in~\Cref{Query2}, (2) finding active vertices in line \ref{Query1 start} - \ref{Query1 end} and \ref{Query1}, and (3) finding the spanning forest with subgraph $G_i$ in~\Cref{Query3}.

\textbf{\Cref{Query2}:}
By~\Cref{corollary alg3property}, the total expected number of queries that used in calling~\Cref{alg ZERCC} is at most 
$\sum_{i = 0}^{6\log \log n} \left(\frac{7}{8}\right)^i n = O(n).$

\noindent
\textbf{Line \ref{Query1 start} - \ref{Query1 end} and \ref{Query1}:} To learn whether a vertex is active takes 1 query. Every time~\Cref{alg SFA} invokes~\Cref{alg ZERCC}, in each connected component, we find one ``active" vertex (if there is any) by iterating over vertices that are not ``inactive" at the end. Once a vertex is ``inactive", it becomes ``inactive" forever. The only vertices that we may query more than once are the representatives, in which case in the $i$th call of~\Cref{alg ZERCC} there are at most $(7/8)^in$ of them. So the total number of queries used to find representatives is at most $n + \sum_{i = 0}^{6\log\log n}(\frac{7}{8})^i{n} = O(n)$.\\

\noindent
\textbf{\Cref{Query3}:}  After $6\log\log n$ calls of~\Cref{alg ZERCC}, the number of connected components of $G_{6\log\log n}$ is at most $(\frac{7}{8})^{6\log \log n} \le \frac{1}{\log n}$. We apply~\Cref{lemma dfs} to find the remaining tree edges deterministically which takes $O(n)$ $\CROSS$ queries.
Summing up the queries from the three parts, we conclude the total expected number of queries is $O(n)$.
%
%
\end{proof}

%% file: section4_det_alg.tex
\section{Deterministic Spanning Forest Algorithm}\label{sec:4}

We describe a polynomial time deterministic algorithm to learn a spanning forest of an undirected {\em unweighted} graph
which makes $O(\frac{n\log n}{\log\log n})$ queries, where $n$ is the number of vertices in $G$. 
The algorithm proceeds in two stages. In the first stage, the algorithm finds ``dense connected components'' in $O(\frac{n\log n}{\log\log n})$ queries
with the guarantee that the number of unknown edges {\em across} the components is ``small''. In the second stage, it learns {\em all} these edges by (i) shrinking all connected components
to be a ``super node'', (ii) using graph reconstruction algorithms to find {\em pseudo-edges} connecting these super-nodes, and then (iii) using ideas from coin-weighing algorithms
to recover the true edges. 

\subsection{Find Dense Connected Components.}
We first describe a deterministic algorithm which makes $O\left(\frac{n\log n}{\log\log n}\right)$ queries and returns a collection $\calC = (C_1, \ldots, C_k)$ of connected
components along with $E'\subseteq E$ which are spanning trees of each $C_i$. The guarantee is that the number of cross-edges $E[\calC]$ defined as $\{(u,v)\in E~:~ u\in C_i, v\in C_j, i\neq j\}$
is small. 

We maintain $A\subseteq V$ to be a set of ``active vertices'' which is initialized to $V$. 
These are the set of vertices from which we haven't begun exploring.

\begin{definition}
	For any vertex $v$, let $\resedge(v)$ be the set of edges between $v$ and the set of vertices $A$.
	Let $\resdeg(v) := |\resedge(v)|$ to be the number of these edges.
\end{definition}
We will use the following subroutines.
\bshouty*


\begin{lemma}~\citep[Paraphrasing][Theorem 3]{BshoutyM11}
There exists a polynomial time algorithm that reconstructs a weighted hidden graph $G = (V, E, w)$ where $w: E \mapsto \mathbb{R}^+$ using $O(\frac{m \log n}{\log m} + m \log \log m)$ $\ADD$ queries. 
\end{lemma}
Using~\Cref{prop:2}, we get
\begin{corollary}\label{lem:BM11}
Given $m \ge n$, there exists a polynomial time algorithm $\BMdet$ that reconstructs a weighted hidden graph $G = (V, E, w)$ where $w: E \mapsto \mathbb{R}^+$ using $O(n + \frac{m \log n}{\log m} + m \log \log m)$ $\CUT$ queries. 
\end{corollary}

\begin{lemma}\label{lem:googoo}
For any vertex $v$ and any subset $A\subseteq V$, the set of edges $\resedge(v)$ can be learnt using $O\left(\frac{\resdeg(v)\cdot \log n}{\log \resdeg(v)}\right)$ queries
via the deterministic algorithm $\Bshouty$ as described in~\Cref{lem:bshouty}.
\end{lemma}

\begin{proof}
	Once $v$ and $A$ are fixed, then one can think of learning $\resedge(v)$ as figuring out the following unknown vector $x\in \{0,1\}^{|A|}$
	which has a coordinate for every vertex in $A$. The $u$th coordinate is $1$ if $(v,u) \in E$ and $0$ otherwise.
	Note that querying $\sum_{u\in S} x_u$ for $S\subseteq A$ is precisely a $\CROSS$ query. The lemma follows from~\Cref{lem:bshouty}.
\end{proof}

Our algorithm proceeds by doing a breadth first search from active vertices only trying to connect to other active vertices only if $\resdeg(v)\geq L = \Theta(\log n/(\log\log n)^2)$. If no such vertex exists, then we the component discovered so far as one component. The full details are described in~\Cref{alg det}.

\begin{algorithm}[ht!]
\caption{$\mathsf{Connected Component Discovery}$}\label{alg det}
\begin{algorithmic}[1]
\Require A graph $G = (V, E)$ with $\CUT$ query access.
\Ensure Return $\calC = (C_1, \ldots, C_k)$, $E'\subseteq E$. Total number of queries: $O\left(\frac{n\log n}{\log \log n}\right)$.

\Comment{Guarantees: (i)	Each $C_i$ is connected using edges of $E'$, (ii) number of cross-edges $|E[\calC]|$
		is $O\left(\frac{n\log n}{(\log \log n)^2}\right)$}
\State $E' \gets \emptyset$; $A \gets V$. 

\State $L \gets \frac{\log n}{(\log \log n)^2}$. 
\State $\calC \gets \emptyset$.
\Comment{Set of connected components} 

\While{$A\neq \emptyset$}	\label{alg:outer-while}
\State Let $x \in A$ be an arbitrary active vertex.

\State Start new connected component $C \gets \{x\}$ and $A\gets A\setminus \{x\}$. 

\State $Q.\add(x)$ \Comment{Start BFS from $x$; $Q$ is a queue} 

\While{$Q\neq \emptyset$} \label{alg:inner-while}
\State $v \gets Q.\remove()$. 

\State $C \gets C \cup \{v\}$. \Comment{Add $v$ to the connected component.}

\State Query $\resdeg(v)$. \label{alg:query-degree}

\If{$\resdeg(v) \ge L$}
\State Use algorithm $\Bshouty$ to learn $\resedge(v)$.  \label{alg:query-cw}
\State $E' = E' \cup \resedge(v)$. 
\For{all $w$ such that $(v,w)\in E'$}
\State $Q.\add(w)$ and $A\gets A\setminus \{w\}$. 
\Comment{Add $w$ to $Q$ and mark it inactive}
\EndFor
\EndIf
\EndWhile

\State Add $C$ to $\calC$. \Comment{$C$ is connected; all vertices $v\in C$ have $\resdeg(v) < L$}.
\EndWhile
\State \Return{$\calC$.}
\end{algorithmic}
\end{algorithm}

\begin{lemma}\label{lem:det1}
	\Cref{alg det} makes $O(\frac{n\log n}{\log \log n})$ $\CUT$ queries and returns $\calC$ with $|E[\calC]| = O\left(\frac{n\log n}{(\log\log n)^2}\right)$.
\end{lemma}

\begin{proof}
	Queries are only made in~\Cref{alg:query-degree} and in~\Cref{alg:query-cw} of~\Cref{alg det}. Amortized, the number of queries made in~\Cref{alg:query-degree} is at most $n$ since we do it whenever the queue is emptied
	and a vertex never re-enters $Q$ because it enters the queue as an inactive vertex. The queries of~\Cref{alg:query-cw} are made only when $\resdeg(v) \geq L$, and so 
	the number of queries is $O\left(\resdeg(v) \cdot \frac{\log n}{\log L}\right)$. We can charge $O\left(\frac{\log n}{\log L}\right) = O\left(\frac{\log n}{\log\log n}\right)$
	to every vertex $w$ that is being added to $C$ (alternately being deemed inactive). Since a vertex is added to $C$ at most once, amortized the total number of such queries is 
	at most $O(\frac{n\log n}{\log \log n})$.

    	Orient all the edges $(u,v)$ of $E[\calC]$ for $u\in C_i$ and $v\in C_j$ from $u$ to $v$ if $i < j$. That is, if $u$ is marked {\em inactive} earlier than $v$.
	Note that the out-degree of every such vertex is $< L$ because when it was made inactive, $\resdeg(v) < L$, and the final out-degree can only go down (maybe some active vertex was added to $C$ as BFS progressed). Since $|E[\calC]|$ is precisely the sum of these out-degrees, the lemma follows.
 
\end{proof}


\subsection{Recover True Edges between Components}

Consider the unknown undirected {\em weighted}/multi graph $H = (\calC, F)$ whose vertices are the components returned by~\Cref{alg det} and
we form $F$ by adding the pair $(C_i,C_j)$ for every $v_i\in C_i$ and $v_j\in C_j$ such that $(v_i,v_j) \in E(G)$. We call the edges in $F$ as ``pseudo-edges''.
We begin with a simple observation.

\begin{lemma}\label{lem:kiki}
	A $\CUT$ query in $H$ can be simulated by a $\CUT$ query in $G$. 
	$H$ is connected if and only if $G$ is connected.
\end{lemma}

\begin{proof}
	Given a subset $\mathcal{T} \subseteq \calC$, note that $\partial_H(T)$ is precisely $\partial_G(\bigcup_{C\in \mathcal{T}} C)$: the weight/number of parallel copies of $(C_i, C_j)$ for $C_i\in \mathcal{T}$ and $C_j\notin \mathcal{T}$ is precisely the number of edges of the form $(v_i,v_j)$ where $v_i\in C_i$ and $v_j\in C_j$. 
	Thus, if $H$ is disconnected, then $G$ is disconnected. If $H$ is connected, then $G$ is connected because every $C_i$ is connected in $G$.
\end{proof}

Using the fact that $H$ can be reconstructed using $O(\frac{m_H \log n}{\log m_H} + m_H\log\log m_H)$ deterministic $\CUT$ queries via~\Cref{lem:BM11} where $m_H$ is the total number of edges (even counting without multiplicity), and using the fact that $m_H \leq nL$, 
we get that $H$ can be reconstructed in $O(\frac{nL\log n}{\log (nL)} + nL\log \log (nL)) = O(nL + nL\log \log n) = O(\frac{n \log n}{\log \log n})$ queries. 
Thus, if we only cared to know if $G$ is connected or not, we would have our $O(n\log n/\log\log n)$ query algorithm by just checking if the reconstructed $H$ is connected or not.

However, the spanning forest of $H$ contains pseudo-edges of the form $(C_i, C_j)$, but we need to return true edges of $G$. We now describe the process to do the same.
We first begin with a simple observation.
\begin{lemma}\label{lem:simp-bins} 
Suppose $(C_i,C_j)$ is a pseudo-edge in $F$. Then, a true edge $(v_i,v_j)$ with $v_i\in C_i$ and $v_j\in C_j$ can be found in $O(\log|C_i| + \log|C_j|)$ $\CUT$ queries
using a binary-search style routine $\BS$.
\end{lemma}

\begin{proof}
	Indeed, for this we don't need the full power of $\CUT$ queries, but just knowing whether a cut is empty or not suffices.
	We take an arbitrary half of $A\subseteq C_i$ and perform a $\CROSS(A, C_j)$ query. If it is $> 0$, then we recurse on $(A, C_j)$
	and if it is $= 0$ we recurse on $(C_i\setminus A, C_j)$. In $\log |C_i|$ such queries, we discover $v_i \in C_i$ which has at least 
	one edge to some vertex in $C_j$. We then repeat the same process by taking an arbitrary half $B\subseteq C_j$ but not performing 
	$\CROSS(\{v_i\}, B)$. In $O(\log |C_j|)$ more queries we would recover the desired edge $(v_i, v_j)$.
\end{proof}
Now fix a spanning forest $\calF$ of $H$. For each $(C_i,C_j) \in \calF$ we want to recover a true edge $(v_i,v_j)$ with $v_i\in C_i$ and $v_j \in C_j$.
If the degree of every vertex $C_i \in \calF$ was $\leq D$, then we could simply use the algorithm in~\Cref{lem:simp-bins} to obtain all such true edges 
with number of queries equaling 
\[
\sum_{(C_i,C_j)\in \calF} O(\log |C_i| + \log |C_j|) \leq D\cdot \sum_{C_i \in \calF} \log |C_i| \underbrace{\leq}_{\textrm{AM-GM}} O(kD\log(n/k)) \leq O(nD)
\]
where $|\calC| = k$. So, if $D \leq \frac{\log n}{\log\log n}$, we would be done. However, $D$ could indeed be as large as $\Theta(n)$. To take care of this, 
we need to handle the ``high-degree'' vertices of $\calF$ separately, again using the fact that $\CUT$ queries give us more power. Here is the observation.
	Let $C_i$ be a vertex with degree $D \geq \frac{\log n}{\log\log n}$ in $\calF$. 
	Let $C_1, \ldots, C_D$ be its neighbors. Now consider the bipartite graph $H'$ whose vertices are the vertices in $C_i$ in one part
	and the vertices $\{C_1, C_2, \ldots, C_j\}$ in the other part, with an edge between $v_i\in C_i$ and $C_j$ iff there exists $v_j\in C_j$ such that
	$(v_i,v_j)\in E(G)$. We now use $\BMdet$ to learn all the edges in $H'$ deterministically. Thus, for all the pseudo-edges $(C_i,C_j)$ incident on $C_i$, we discover {\em one} endpoint (that lying in $C_i$) of $\geq D$ many true edges. This is how we take care of high-degree vertices. 
	The full algorithm details are given in~\Cref{alg join}.

\begin{algorithm}
	\caption{$\mathsf{Join Connected Components}$}\label{alg join}
	\begin{algorithmic}[1]
	\Require Unknown undirected graph $G = (V, E)$ with $\CUT$ query access; $\calC = (C_1, \ldots, C_k)$; $E'\subseteq E$ such that each $C_i$ is connected using edges in $E'$.
	\Ensure Return spanning forest of $G$. Total number of queries: $O\left(\frac{n\log n}{\log \log n}\right)$.
	
	\State Learn the pseudo-graph $H = (\calC, F)$ using $\BMdet$ in~\Cref{lem:BM11}. \label{alg:bm1}
	
	\State $\calF$ be an arbitrary spanning forest of $H$. 
	
	\State $E'' \gets \emptyset$.
	
	\For{$C_i \in \calF$ with degree $D > \frac{\log n}{\log\log n}$}
		
		\State Let $C_1, \ldots, C_D$ be neighbors of $C_i$ in $\calF$.
		
		\State Use $\BMdet$ to learn bipartite graph $H_i := (C_i, \{C_1, \ldots, C_D\})$. \label{alg:bm2}
		
		\State Replace $\{C_i\}$ in $\calF$ with $\{v: v\in C_i\}$ and edges $E'\cap E[C_i]$.
		
		\State Add minimal collection of edges from $H_i$ to ensure $C_i \cup \{C_1, \ldots, C_D\}$ is connected, to $E''$.
	\EndFor
	
	\Comment{Now all $C_i$'s in $\calF$ have degree at most $ \frac{\log n}{\log\log n}$.}
	
	\State \For{all remaining $(C_i,C_j)$ or $(C_i, v)$ pseudo-edges of $\calF$}
	
		\State Use $\BS$ to learn true edge $(v_i,v_j)$ or $(v_i, v)$ with $v_i\in C_i$ and $v_j\in C_j$, and add to $E''$.
	
	\EndFor
	
	\State \Return{$E' \cup E''$.}
	
\end{algorithmic}
\end{algorithm}

\begin{theorem}\label{thm:det-thm}
\Cref{alg det} and \Cref{alg join} is a deterministic algorithm that returns a maximal spanning forest of $G$ in $O\left(\frac{n\log n}{\log\log n}\right)$ queries.
\end{theorem}

\begin{proof}
	We have already analyzed the query complexity of~\Cref{alg det} in~\Cref{lem:det1}.~\Cref{lem:kiki} implies that the number of edges in $H$ is $O\left(\frac{n\log n}{(\log\log n)^2}\right)$.
	The number of queries made in~\Cref{alg:bm1} of~\Cref{alg join} is therefore, as argued in the beginning of this subsection, $O(\frac{n \log n}{\log \log n})$.
	The number of queries made in~\Cref{alg:bm2} of~\Cref{alg join} is $O(\frac{m_i\log n_i}{\log m_i} + m_i\log\log m_i)$ where $m_i$ is the number of edges in the bipartite graph $H_i$ and $n_i$ is
	the number of vertices. The second term, summed over all calls, and using the fact that $\sum_i m_i = O(nL)$, would together be at most $O(n\log n/\log\log n)$.
	The first term, if $m_i \geq n_i$ would give $O(1)$ per edge, which again amortized would be $O(nL)$. If $m_i \ll n_i$, then using $m_i \geq \log n/\log\log n$, 
	we would still pay $O(\frac{\log n}{\log\log n})$ per true edge added to $E''$ at that round. Since the total number of true edges in $E''$ is at most $n-1$, we 
	would in all pay $O(n\log n/\log\log n)$. 
\end{proof}

%% file: section5_conclusion.tex
\section{Conclusion}

We give an optimal Las Vegas algorithm style that recovers a maximal spanning forest of an unknown, weighted, undirected graph using $O(n)$ queries in expectation. The algorithm uses~\cite{Apers}'s framework, and extends the $O(n)$ zero-error algorithm in unweighted graph to weighted case. It also closes the gap between the lower bound result of $\Omega(n)$ in~\cite{AuzaL21} and previous upper bound of $O(n \log n)$. A key ingredient of the improvement is that we circumvent the degree estimation, which is possible not an $O(1)$-query operation in weighted graphs given the results of~\cite{chakraborty2022support}. We do this by removing the high degree representatives step by step using random sampling with a sample rate increasing from close to zero to one. 

For unweighted graphs, we give an $O(n\log n/\log\log n)$-query poly-time deterministic algorithm slightly improving the state-of-the-art. Our algorithm doesn't work with weights yet again due to the need for knowing degree (\Cref{alg:query-degree} in~\Cref{alg det}), and we don't know how to bypass it without randomization. However, the more interesting question is to obtain
$O(n)$ query deterministic algorithms even for unweighted graphs, or prove its impossibility. 

%% file: sectionA_missing_proofs.tex
\section{Missing Proofs}\label{Appendix:A}

\begin{proposition} ~\citep[from][Claim 3]{ChakrabartyL23} \label{claim 3} 
	Let $t_1, t_2, \cdots, t_j$ be integers at least 2. If $\sum_{i = 1}^{j} t_i = t$ where $\frac{n}{\log n} \le t \le n$ and $j \le \log n$, then $\sum_{i = 1}^{j} \frac{t_i}{\log t_i} \le \frac{4t}{\log t} = O(\frac{t}{\log t})$.
\end{proposition}

\begin{proof}
Partition $\{t_i\}$ into sets $P_1 = \{t_i | t_i \in [2, t/\log^2 t)\}, P_2 = \{t_i | t_i \in [t/\log^2 t, t]\}$. 
For $t_i \in P_1$, note $|P_1| \le \log n \le 2\log t$, it follows that $$\sum_{i \in P_1} \frac{i}{\log i} \le \sum_{i \in P_1} i \le 2\log t \cdot \frac{t}{\log^2 t}  = \frac{2t}{\log t}.$$
For $t_i \in P_2$, since $t_i > \sqrt{t}, \log t_i \ge \frac{1}{2}\log t$. So $t_i/\log t_i \le t_i/\frac{1}{2} \log t = \frac{2t_i}{\log t}$. For those $t_i$, $$\sum_{i \in P_2} \frac{i}{\log i} \le \sum_{i \in P_2} \frac{2i}{\log t} \le \frac{2t}{\log t}.$$
Combine the two inequalities, $$\sum_{i \in P} \frac{i}{\log i} = \sum_{i \in P_1} \frac{i}{\log i} + \sum_{i \in P_2} \frac{i}{\log i} \le \frac{4t}{\log t} = O(\frac{t}{\log t}).$$
\end{proof}

\begin{theorem}(\Cref{lemma eochoi} in the main body)~\citep[Follows from][Theorem 1]{Choi13} . 
	There exists an adaptive, randomized algorithm $\GRtwo$ which takes input a bipartite graph $G = (U, V)$ on $n$ vertices and $m$ edges such that $U$ and $V$ are known but $m$ is unknown to the 
	algorithm, and either reconstructs the edges of $G$ along with their weights, or aborts. The probability that the algorithm aborts or makes more than $\frac{C_{\GRtwo}m\log n}{\log m}$ $\CROSS$ queries, for some absolute constant $C_{\GRtwo} > 0$, is at most $O(\frac{\log m}{m})$.
\end{theorem}

\begin{proof}
	We run $\Choi$ in~\Cref{lemma choi} with number of edge guesses $2^1, 2^2, \cdots$ growing in powers of $2$. Let $C_{\GRC}$ be the hidden constant in the query complexity of~\Cref{lemma choi}. After guessing $2^i$, $\Choi$ terminates with $C_{\GRC} \cdot \frac{2^i \log n}{\log 2^i}$ queries. If $\Choi$ returns recovered edges, we check 
    \begin{enumerate}[noitemsep]
        \item If the edges indeed exist in $G$.
        \item If all the edges are recovered.
    \end{enumerate}
	
	For a guess $2^i$, checking the recovered edges are indeed in $G$ takes at most $2^i$ queries. Checking if all edges are recovered takes 1 $\CROSS$ query because $G$ is bipartite and it takes 1 query to get the sum of all edge weights. One then check if the sum of the recovered edge weights equal to the sum of actual edge weights.
	
	When the guess $2^i \ge m$, running $\Choi$ with guess $2^i$ fails with probability $O(\frac{\log m}{m})$. Hence the probability that the algorithm stops at guess $2^{i^*}$ where $i^*$ is the smallest $i$ such that $2^i \ge m$ is $1 - O(\frac{\log m}{m})$. Conditioned on the algorithm stops on the guess $2^{i^*}$, we show that the total number of queries is $\sum_{i = 1}^{i^*}C_{\GRC} \cdot \frac{2^i\log n}{\log{2^i}} \le 16C_{\GRC} \cdot \frac{m\log n}{\log m}$.

\begin{align*}
    \sum_{i = 1}^{i^*} \left(C_{\GRC} \cdot \frac{2^i\log n}{\log 2^i} \right) &\le 4C_{\GRC} \cdot \left(\frac{2^{i^* + 1}\log n}{\log 2^{i^* + 1}} \right) \quad \text{(by Proposition~\ref{claim 3})}
    \\ &\le 16C_{\GRC} \cdot \left(\frac{2^{i^* - 1}\log n}{\log 2^{i^* + 1}} \right)
    \\ &\le 16C_{\GRC} \cdot \left(\frac{m\log n}{\log m} \right)
\end{align*}

Set $C_{\GRtwo}$ to be $16C_{\GRC}$ and we are done. 

 The proof of~\Cref{lemma eobm} is similar to the proof of~\Cref{lemma eochoi}.
\end{proof}